\documentclass[twocolumn]{article}
\usepackage{tikz}
\usetikzlibrary{calc}

\input{toffoli.sty}

\title{Quantum circuits of $T$-depth one}

\author{Peter Selinger\\
  Department of Mathematics and Statistics, Dalhousie University}

\date{\begin{minipage}{0.9\textwidth}
  \normalsize We give a Clifford+$T$ representation of the Toffoli gate of
  $T$-depth 1, using four ancillas. More generally, we describe a
  class of circuits whose $T$-depth can be reduced to 1 by using
  sufficiently many ancillas.  We show that the cost of adding an
  additional control to any controlled gate is at most 8 additional
  $T$-gates, and $T$-depth 2. We also show that the circuit $THT$ does
  not possess a $T$-depth 1 representation with an arbitrary number of
  ancillas initialized to $\ket{0}$.
\end{minipage}}

\begin{document}
\maketitle

\section{Introduction}

It is known that the gates of the Clifford group, together with the
single-qubit non-Clifford gate
\[ T = \zmatrix{cc}{1&0\\0&e^{i\pi/4}},
\]
form a good universal gate set for fault-tolerant quantum computation
{\cite{Buhrman-Cleve-etal}}. The decomposition of arbitrary gates into
this Clifford+$T$ set, either exactly or to within some given
accuracy $\epsilon$, is an important problem
{\cite{Kliuchnikov-et-al}}. It is often desirable to find
decompositions that are optimal with respect to a given cost
function. The exact cost function used is application dependent; some
possibilities are: the total number of gates; the total number of
$T$-gates; the circuit depth; and/or the number of ancillas used.

Amy et al.~{\cite{AMMR12}} recently proposed {\em $T$-depth} as a cost
function. The idea is to count the number of {\em $T$-stages} in a
circuit, rather than the number of $T$-gates. A $T$-stage is a group
of one or more $T$- and/or $T\da$-gates on distinct qubits that can be
performed simultaneously. Note that, for the purpose of computing
$T$-count or $T$-depth, the gates $T$ and $T\da$ can be treated
interchangeably, due to the identity $T\da=TS\da$.

To illustrate the concept of $T$-depth, consider the standard
decomposition of the Toffoli gate into the Clifford+$T$ set, as
given in {\cite{Nielsen-Chuang}}:
\begin{equation}\label{eqn-nc}
  \m{\begin{qcircuit}[scale=0.5]
    \grid{2}{0,1,2};
    \controlled{\notgate}{1,0}{1,2};
  \end{qcircuit}
  } =
  \m{\begin{qcircuit}[scale=0.5]
    \grid{14.2}{0,1,2};
    \gate{$H$}{1,0};
    \controlled{\notgate}{2,0}{1};
    \gate{$T\da$}{3,0};
    \controlled{\notgate}{4,0}{2};
    \gate{$T$}{5,0};
    \controlled{\notgate}{6,0}{1};
    \gate{$T\da$}{7,0};
    \controlled{\notgate}{8,0}{2};
    \gate{$T$}{9,0};
    \gate{$T\da$}{9,1};
    \gate{$H$}{10.2,0};
    \controlled{\notgate}{10.2,1}{2};
    \gate{$T\da$}{11.2,1};
    \controlled{\notgate}{12.2,1}{2};
    \gate{$S$}{13.2,1};
    \gate{$T$}{13.2,2};
  \end{qcircuit}}
\end{equation}
This decomposition has $T$-count $7$, and in the exact form written,
it has $T$-depth $6$, because the fourth and fifth $T$-gates form a
single $T$-stage. Using trivial commutations, the circuit
(\ref{eqn-nc}) can easily be reduced to $T$-depth $4$:
\begin{equation}\label{eqn-nc4}
  \m{\begin{qcircuit}[scale=0.5]
    \grid{2}{0,1,2};
    \controlled{\notgate}{1,0}{1,2};
  \end{qcircuit}
  } =
  \m{\begin{qcircuit}[scale=0.5]
    \grid{13}{0,1,2};
    \gate{$H$}{1,0};
    \controlled{\notgate}{2,0}{1};
    \gate{$T\da$}{3,0};
    \controlled{\notgate}{4,0}{2};
    \gate{$T$}{5,0};
    \controlled{\notgate}{6,0}{1};
    \gate{$T\da$}{7,0};
    \gate{$T\da$}{7,1};
    \controlled{\notgate}{8,0}{2};
    \gate{$T$}{10,0};
    \gate{$H$}{12,0};
    \controlled{\notgate}{9,1}{2};
    \gate{$T\da$}{10,1};
    \controlled{\notgate}{11,1}{2};
    \gate{$S$}{12,1};
    \gate{$T$}{10,2};
  \end{qcircuit}}
\end{equation}
Amy et al.~{\cite{AMMR12}} further improved the $T$-depth of the
Toffoli gate to $3$, using the following circuit. They conjecture that
for circuits without ancillas, this $T$-depth is optimal.
\begin{equation}\label{eqn-ammr}
  \scalebox{0.8}{\mp{0.55}{\begin{qcircuit}[scale=0.5]
    \grid{2}{0,1,2};
    \controlled{\notgate}{1,0}{1,2};
  \end{qcircuit}}
  } =
  \m{\begin{qcircuit}[scale=0.5]
    \gridx{-.2}{11}{0,1,2};
    \gate{$H$}{0.8,0};
    \gate{$T$}{2,0};
    \gate{$T$}{2,1};
    \gate{$T\da$}{2,2};
    \controlled{\notgate}{3,1}{2};
    \controlled{\notgate}{4,2}{0};
    \gate{$T\da$}{5,2};
    \controlled{\notgate}{5,0}{1};
    \controlled{\notgate}{6,2}{1};
    \gate{$T$}{7,0};
    \gate{$T\da$}{7,1};
    \gate{$T\da$}{7,2};
    \controlled{\notgate}{8,2}{0};
    \controlled{\notgate}{9,0}{1};
    \gate{$S$}{9,2};
    \gate{$H$}{10,0};
    \controlled{\notgate}{10,1}{2};
  \end{qcircuit}}
\end{equation}
The purpose of this note is to show that, with the use of ancillas,
the $T$-depth of the Toffoli gate, and of many (but not all) other
circuits, can be reduced to 1. This may be useful in quantum computing
architectures where $T$-gates are expensive and ancillas are cheap.

\section{A $T$-depth one representation of the Toffoli gate}
\label{sec-toffoli}

\begin{figure*}
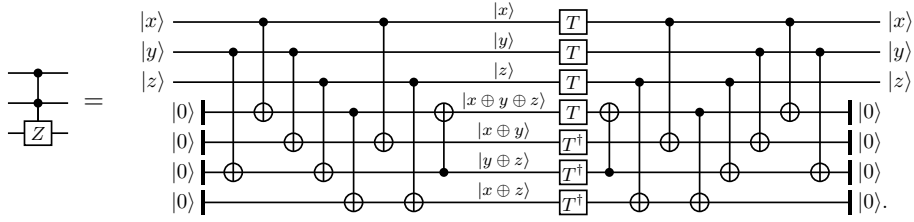

\[
  \m{\begin{qcircuit}[scale=0.5]
      \grid{2}{0,1,2};
      \controlled{\gate{$Z$}}{1,0}{1,2};
    \end{qcircuit}
  } =
  \m{\begin{qcircuit}[scale=0.5]
      \gridx{0}{23.5}{4,5,6};
      \gridx{1}{22.5}{0,1,2,3};
      \leftlabel{$\ket x$}{0,6};
      \leftlabel{$\ket y$}{0,5};
      \leftlabel{$\ket z$}{0,4};
      \init{$\ket 0$}{1,0};
      \init{$\ket 0$}{1,1};
      \init{$\ket 0$}{1,2};
      \init{$\ket 0$}{1,3};
      \controlled{\notgate}{3,3}{6};
      \controlled{\notgate}{4,2}{5};
      \controlled{\notgate}{5,1}{4};
      \controlled{\notgate}{6,0}{3};
      \controlled{\notgate}{7,2}{6};
      \controlled{\notgate}{2,1}{5};
      \controlled{\notgate}{8,0}{4};
      \controlled{\notgate}{9,3}{1};
      \wirelabel{$\ket{x}$}{11,6};
      \wirelabel{$\ket{y}$}{11,5};
      \wirelabel{$\ket{z}$}{11,4};
      \wirelabel{$\ket{x\oplus y\oplus z}$}{11,3};
      \wirelabel{$\ket{x\oplus y}$}{11,2};
      \wirelabel{$\ket{y\oplus z}$}{11,1};
      \wirelabel{$\ket{x\oplus z}$}{11,0};
      \gate{$T$}{13.3,6};
      \gate{$T$}{13.3,5};
      \gate{$T$}{13.3,4};
      \gate{$T$}{13.3,3};
      \gate{$T\da$}{13.3,2};
      \gate{$T\da$}{13.3,1};
      \gate{$T\da$}{13.3,0};
      \controlled{\notgate}{20.5,3}{6};
      \controlled{\notgate}{19.5,2}{5};
      \controlled{\notgate}{18.5,1}{4};
      \controlled{\notgate}{17.5,0}{3};
      \controlled{\notgate}{16.5,2}{6};
      \controlled{\notgate}{21.5,1}{5};
      \controlled{\notgate}{15.5,0}{4};
      \controlled{\notgate}{14.5,3}{1};
      \rightlabel{$\ket x$}{23.5,6};
      \rightlabel{$\ket y$}{23.5,5};
      \rightlabel{$\ket z$}{23.5,4};
      \term{$\ket 0$.}{22.5,0};
      \term{$\ket 0$}{22.5,1};
      \term{$\ket 0$}{22.5,2};
      \term{$\ket 0$}{22.5,3};
    \end{qcircuit}
  }
\]
\caption{$T$-depth 1 representation of the Toffoli gate}\label{fig-1}
\end{figure*}
Recall that the Clifford group for any number of qubits is generated
by the Hadamard gate $H$, the phase gate $S=T^2$, the controlled
not-gate, and unit scalars. As usual, we write $X$, $Y$, and $Z$ for
the Pauli operators.
\[ H = \frac{1}{\sqrt{2}}\zmatrix{cc}{1&1\\1&-1}, \sep
   S = \zmatrix{cc}{1&0\\0&i},
\]
\[
   X = \zmatrix{cc}{0&1\\1&0}, \sep
   Y = \zmatrix{cc}{0&-i\\i&0}, \sep
   Z = \zmatrix{cc}{1&0\\0&-1}
\]
The Toffoli gate is a doubly-controlled not-gate. It is equivalent to
a doubly-controlled $Z$-gate via a basis change:
\begin{equation}\label{eqn-xhzh}
  \m{\begin{qcircuit}[scale=0.45]
      \grid{2}{0,1,2};
      \controlled{\notgate}{1,0}{1,2};
    \end{qcircuit}
  } =
  \m{\begin{qcircuit}[scale=0.45]
      \grid{5}{0,1,2};
      \gate{$H$}{1,0};
      \controlled{\gate{$Z$}}{2.5,0}{1,2};
      \gate{$H$}{4,0};
    \end{qcircuit}.
  }
\end{equation}
Now consider a computational basis state $\ket{xyz}$, where
$x,y,z\in\s{0,1}$.  The effect of the doubly-controlled $Z$-gate is to
map $\ket{xyz}$ to $(-1)^{xyz}\ket{xyz}$. Let us write ``$\oplus$''
for modulo-2 addition in $\s{0,1}$, and ``$+$'' and ``$-$'' for the
usual addition and subtraction of integers. We then have the following
inclusion-exclusion style formula for $x,y,z\in\s{0,1}$:
\begin{equation}\label{eqn-inc-exc}
 4xyz = x + y + z - (x\oplus y) - (y\oplus z) - (x\oplus z) + (x\oplus y\oplus z).
\end{equation}
This is easy to prove by case distinction, or algebraically using
$x\oplus y = x+y-2xy$. Now let $\omega=(-1)^{1/4}=e^{i\pi/4}$.
From (\ref{eqn-inc-exc}), we have
\begin{equation}\label{eqn-omega}
  \begin{split}
  &(-1)^{xyz} = \omega^{4xyz} 
  \\&\quad= \omega^{x}\; \omega^{y}\; \omega^{z}\; (\omega\da)^{x\oplus y}\; (\omega\da)^{y\oplus z}\; (\omega\da)^{x\oplus z}\; \omega^{x\oplus y\oplus z}.
\end{split}
\end{equation}
Note that $T\ket{x} = \omega^x\ket{x}$, and therefore, the
doubly-controlled $Z$-gate can be implemented by applying $T$-gates to
qubits in states $\ket{x}$, $\ket{y}$, $\ket{z}$, and $\ket{x\oplus
  y\oplus z}$, and $T\da$-gates to qubits in states $\ket{x\oplus
  y}$, $\ket{y\oplus z}$, and $\ket{x\oplus z}$. This can be done in
any order, or even in parallel, using four ancillas, as shown in
Figure~\ref{fig-1}.
Combining this with (\ref{eqn-xhzh}), we obtain a representation of
the Toffoli gate of $T$-depth 1 and overall depth~7.

\begin{remark}
  It is interesting to note that the decompositions of Nielsen and
  Chuang (\ref{eqn-nc}) and Amy et al.~(\ref{eqn-ammr}) follow
  precisely the same pattern, i.e., they can both be seen to be direct
  implementations of (\ref{eqn-omega}). The only difference is that in
  each of the circuits, one of the $T$-gates has been needlessly
  decomposed into $T\da$ and $S$.
\end{remark}

\section{An application to multiply-con\-trolled gates}

Consider a doubly-controlled $(-iZ)$-gate:
\begin{equation}\label{eqn-ps-gate}
  \m{\begin{qcircuit}[scale=0.5]
      \gridx{-.3}{2.3}{0,1,2};
      \controlled{\widegate{$-iZ$}{.75}}{1,0}{1,2};
    \end{qcircuit}
  } = \!\!\!
  \m{\begin{qcircuit}[scale=0.5]
      \grid{3}{0,1,2};
      \leftlabel{$\ket x$}{0,2};
      \leftlabel{$\ket y$}{0,1};
      \leftlabel{$\ket z$}{0,0};
      \controlled{\gate{$Z$}}{1,0}{1,2};
      \controlled{\gate{$S\da$}}{2,1}{2,2};
    \end{qcircuit}.
  }
\end{equation}
The doubly-controlled $Z$-gate is a diagonal gate whose effect is
given by (\ref{eqn-omega}). The controlled $S\da$-gate is a diagonal
gate whose effect is given by
$(-i)^{xy} = (\omega\da)^{x}\;(\omega\da)^{y}\;\omega^{x\oplus y}$. It follows
that the combined effect of the two gates is
\begin{equation}
  (-1)^{xyz}\;(-i)^{xy} = \omega^{z}\; (\omega\da)^{y\oplus z}\; (\omega\da)^{x\oplus z}\;
  \omega^{x\oplus y\oplus z},
\end{equation}
which therefore requires a $T$-count of only~4. Using one
ancilla, this can be achieved with $T$-depth~1 and overall depth~5:
\begin{equation}\label{eqn-tri-z}
  \m{\begin{qcircuit}[scale=0.5]
      \gridx{-.3}{2.3}{0,1,2};
      \controlled{\widegate{$-iZ$}{.75}}{1,0}{1,2};
    \end{qcircuit}
  } \ = \!\!\!
  \m{\begin{qcircuit}[scale=0.5]
      \gridx{1}{14.5}{4,5,6};
      \gridx{2}{13.5}{3};
      \leftlabel{$\ket x$}{1,6};
      \leftlabel{$\ket y$}{1,5};
      \leftlabel{$\ket z$}{1,4};
      \init{$\ket 0$}{2,3};
      \controlled{\notgate}{2,5}{4};
      \controlled{\notgate}{4,6}{4};
      \controlled{\notgate}{3,3}{6};
      \controlled{\notgate}{5,3}{5};
      \wirelabel{$\ket{x\oplus z}$}{7,6};
      \wirelabel{$\ket{y\oplus z}$}{7,5};
      \wirelabel{$\ket{z}$}{7,4};
      \wirelabel{$\ket{x\oplus y\oplus z}$}{7,3};
      \gate{$T\da$}{9.3,6};
      \gate{$T\da$}{9.3,5};
      \gate{$T$}{9.3,4};
      \gate{$T$}{9.3,3};
      \controlled{\notgate}{13.5,5}{4};
      \controlled{\notgate}{11.5,6}{4};
      \controlled{\notgate}{12.5,3}{6};
      \controlled{\notgate}{10.5,3}{5};
      \term{$\ket 0$.}{13.5,3};
      \rightlabel{$\ket x$}{14.5,6};
      \rightlabel{$\ket y$}{14.5,5};
      \rightlabel{$\ket z$}{14.5,4};
    \end{qcircuit}
  }
\end{equation}
Alternatively, one can find an implementation that uses no ancilla. It
uses fewer overall gates, but has $T$-depth~2 and overall depth~7:
\begin{equation}\label{eqn-tri-z-alt}
  \m{\begin{qcircuit}[scale=0.5]
      \gridx{-.3}{2.3}{0,1,2};
      \controlled{\widegate{$-iZ$}{.75}}{1,0}{1,2};
    \end{qcircuit}
  } \ =\!\!\!
  \m{\begin{qcircuit}[scale=0.5]
      \gridx{1}{14.5}{4,5,6};
      \leftlabel{$\ket x$}{1,6};
      \leftlabel{$\ket y$}{1,5};
      \leftlabel{$\ket z$}{1,4};
      \controlled{\notgate}{2,5}{4};
      \controlled{\notgate}{3,6}{5};
      \wirelabel{$\ket{z}$}{5,4};
      \wirelabel{$\ket{y\oplus z}$}{5,5};
      \wirelabel{$\ket{x\oplus y\oplus z}$}{5,6};
      \gate{$T$}{7.3,6};
      \gate{$T\da$}{7.3,5};
      \gate{$T$}{7.3,4};
      \controlled{\notgate}{8.5,5}{4};
      \controlled{\notgate}{9.5,6}{5};
      \wirelabel{$\ket{x\oplus z}$}{11,6};
      \gate{$T\da$}{12.5,6};
      \controlled{\notgate}{13.5,6}{4};
      \rightlabel{$\ket x$}{14.5,6};
      \rightlabel{$\ket y$}{14.5,5};
      \rightlabel{$\ket z$.}{14.5,4};
    \end{qcircuit}
  }
\end{equation}
We also have
\begin{equation}\label{eqn-tri-not}
  \m{\begin{qcircuit}[scale=0.5]
      \gridx{-.3}{2.3}{0,1,2};
      \controlled{\widegate{$-iX$}{.75}}{1,0}{1,2};
    \end{qcircuit}
  } =
  \m{\begin{qcircuit}[scale=0.5]
      \gridx{-1.5}{3.5}{0,1,2};
      \gate{$H$}{-0.5,0};
      \controlled{\widegate{$-iZ$}{.75}}{1,0}{1,2};
      \gate{$H$}{2.5,0};
    \end{qcircuit}
  } =
  \m{\begin{qcircuit}[scale=0.5]
      \grid{3}{0,1,2};
      \controlled{\notgate}{1,0}{1,2};
      \controlled{\gate{$S\da$}}{2,1}{2,2};
    \end{qcircuit}.
  }
\end{equation}
Suppose we have a Clifford+$T$-representation of some
controlled quantum gate $G$, and we wish to obtain an efficient
Clifford+$T$-representation of a doubly-controlled $G$-gate. Using
(\ref{eqn-tri-z}), (\ref{eqn-tri-not}), and (\ref{eqn-ccg}), the cost
of doing so is at most 8 additional $T$-gates, increasing the
$T$-depth by at most 2, and the overall depth by at most 14, using 2
ancillas:
\begin{equation}\label{eqn-ccg}
  \m{\begin{qcircuit}[scale=0.45]
      \gridx{-.5}{2.5}{0,1,2};
      \leftlabel{$\ket x$}{-.5,2};
      \leftlabel{$\ket y$}{-.5,1};
      \leftlabel{$\ket\phi$}{-.5,0};
      \multiwire{0,0};
      \controlled{\gate{$G$}}{1,0}{2,1};
      \multiwire{2,0};
    \end{qcircuit}
  } =\!\!\!
  \m{\begin{qcircuit}[scale=0.45]
      \gridx{-.6}{6.6}{0,2,3};
      \gridx{.4}{5.6}{1};
      \leftlabel{$\ket x$}{-.6,3};
      \leftlabel{$\ket y$}{-.6,2};
      \leftlabel{$\ket\phi$}{-.6,0};
      \init{$\ket 0$}{.4,1};
      \controlled{\widegate{$-iX$}{.75}}{1.7,1}{3,2};
      \multiwire{1,0};
      \controlled{\gate{$G$}}{3,0}{1};
      \multiwire{5,0};
      \controlled{\widegate{$iX$}{.75}}{4.3,1}{3,2};
      \term{$\ket 0$}{5.6,1};
    \end{qcircuit}.
  }
\end{equation}
Note that the cost of the additional control, in terms of the overall
gate count, is 28 (2 times 12 gates from (\ref{eqn-tri-z}) and 2 times
2 Hadamard gates from (\ref{eqn-tri-not})). This can be reduced to 26
by leaving the ancilla in (\ref{eqn-tri-z}) in state $\ket{x}$ instead
of $\ket{0}$; however, doing so requires carrying this ancilla during
the computation of $G$, which may involve a tradeoff.

If (\ref{eqn-tri-z-alt}) is used instead of (\ref{eqn-tri-z}), the
overall gate count cost of (\ref{eqn-ccg}) decreases to 22, and the
ancilla use to~1. However, the depth and $T$-depth cost increase to~18
and~4, respectively.

\begin{remark}
  The above construction can be iterated to add $n$ additional
  controls to a controlled gate at the cost of $T$-count $8n$ and
  $T$-depth $2\floor{\log_2{n+1}}$. The logarithm in the expression
  for $T$-depth arises because a pair of $T$-stages is sufficient to
  {\em double} the number of controls, as shown here for $n=3$:
  \begin{equation}
    \m{\begin{qcircuit}[scale=0.45]
        \gridx{-.5}{2.5}{0,1,2,3,4};
        \multiwire{0,0};
        \controlled{\gate{$G$}}{1,0}{4,3,2,1};
        \multiwire{2,0};
      \end{qcircuit}
    } = \!\!\!
    \m{\begin{qcircuit}[scale=0.45]
        \gridx{-2.4}{8.4}{0,4,3,6,7};
        \gridx{-.9}{6.9}{2,5};
        \gridx{-.9}{6.9}{1};
        \init{$\ket 0$}{-.9,5};
        \init{$\ket 0$}{-.9,2};
        \init{$\ket 0$}{-.9,1};
        \controlled{\widegate{$-iX$}{.75}}{.4,5}{6,7};
        \controlled{\widegate{$-iX$}{.75}}{.4,2}{4,3};
        \controlled{\widegate{$-iX$}{.75}}{1.7,1}{5,2};
        \multiwire{1,0};
        \controlled{\gate{$G$}}{3,0}{1};
        \multiwire{5,0};
        \controlled{\widegate{$iX$}{.75}}{4.3,1}{5,2};
        \controlled{\widegate{$iX$}{.75}}{5.6,2}{4,3};
        \controlled{\widegate{$iX$}{.75}}{5.6,5}{6,7};
        \term{$\ket 0$}{6.9,1};
        \term{$\ket 0$}{6.9,5};
        \term{$\ket 0$}{6.9,2};
      \end{qcircuit}.
    }
  \end{equation}
  For example, this yields an implementation of a triply-controlled
  not-gate with $T$-count $15$ and $T$-depth $3$ (7 $T$-gates for the
  Toffoli gate, and 8 $T$-gates for the additional control); or a
  quintuply-controlled not-gate with $T$-count $31$ and $T$-depth
  $5$. It is not currently known whether any of these $T$-counts or
  depths are optimal.
\end{remark}

\begin{remark}
  Because the $T$-gate is diagonal with $T\ket0=\ket0$, it can be
  regarded as a controlled gate, namely, a controlled global phase
  change. Therefore, we can use the above procedure to implement a
  controlled $T$-gate with $T$-count 9 as follows:
  \begin{equation}\label{eqn-cT}
  \m{\begin{qcircuit}[scale=0.45]
      \gridx{-.5}{2.5}{1,2};
      \controlled{\gate{$T$}}{1,1}{2};
    \end{qcircuit}
  } = \!\!\!
  \m{\begin{qcircuit}[scale=0.45]
      \gridx{-1.1}{7.1}{2,3};
      \gridx{.4}{5.6}{1};
      \init{$\ket 0$}{.4,1};
      \controlled{\widegate{$-iX$}{.75}}{1.5,1}{3,2};
      \gate{$T$}{3,1};
      \controlled{\widegate{$iX$}{.75}}{4.5,1}{3,2};
      \term{$\ket 0$.}{5.6,1};
    \end{qcircuit}
  }
  \end{equation}
  Using (\ref{eqn-tri-z}), we obtain $T$-depth~3, depth~15, and gate
  count~29 with two ancillas. As before, by leaving the ancilla of
  (\ref{eqn-tri-z}) in state $\ket x$ instead of $\ket 0$, the gate
  count can be reduced to~27. Alternatively, using
  (\ref{eqn-tri-z-alt}), we obtain $T$-depth~5, depth~19, and gate
  count~27 with one ancilla. Except for slightly improved overall gate
  counts, these results are the same as those in {\cite{AMMR12}}.
\end{remark}

\section{$T$-depth one representation of almost classical circuits}

It is straightforward to generalize the construction of
Section~\ref{sec-toffoli} to circuits built up from $T$ and {\em
  almost classical} gates.

\begin{definition}
  A unitary operator is {\em classical} if it is given by a
  permutation of computational basis states, and {\em diagonal} if its
  matrix representation is diagonal in the computational basis. Let us
  call an operator {\em almost classical} if it can be written as a
  product of a classical operator and a diagonal operator.
\end{definition}

The almost classical operators obviously form a group. Of the 24
single-qubit Clifford operators (taken modulo global phase), exactly 8
are almost classical; they form the subgroup generated by $S$ and $X$.

\begin{definition}
  Let $\PCG$ be a set of gates. We say that a circuit is
  {\em $\PCG+T$-representable} if it can be built with gates from
  $\PCG\cup\s{T}$ and their inverses. We say that such a circuit has
  {\em $T$-depth $n$ (relative to $\PCG$)} if it can be written using only
  gates from $\PCG$ and $n$ $T$-stages.
\end{definition}

\begin{theorem}\label{thm-almost}
  Let $\PCG$ be any set of almost classical gates, containing the
  controlled not-gate. Using ancillas, any $\PCG+T$-representable
  $n$-qubit circuit can be written of $T$-depth 1 (relative to
  $\PCG$).
\end{theorem}

\begin{proof}
  The proof idea is simple. Each $T$-gate in the circuit is a
  $\pi/4$ phase change conditioned on some boolean combination of the
  inputs. Intuitively, one may copy each such boolean condition to an
  ancilla, execute all $T$-gates in parallel, uncompute the ancillas,
  and finally re-compute the output.

  The formal proof proceeds by induction on circuits. For each
  $\PCG+T$-representable $n$-qubit circuit $A$, we will by induction
  construct $\PCG+T$-representable circuits $A_1$ and $A_2$ such that
  $A_1$ is diagonal and has $T$-depth at most 1, $A_2$ has $T$-depth
  0, and $A = A_2\cp A_1$.

  The base case occurs when $A=I$ is the identity circuit. In this
  case, we can let $A_1=A_2=I$, and there is nothing to show.

  For the induction step, suppose $A$ is of the form $A'\cp G$, where
  $G$ is a single gate. By induction hypothesis, there is a
  decomposition $A' = A'_2\cp A'_1$ satisfying the above conditions.

  \begin{itemize}
  \item Case 1: $G$ is not equal to $T$ or $T\da$. In this case, we
    let $A_1 = G\da\cp A'_1\cp G$ and $A_2 = A'_2\cp G$. Then
    trivially, $A = A_2\cp A_1$, and $A_1$ and $A_2$ have the required
    $T$-depths. Moreover, since $G$ is almost classical, $A_1$ is
    diagonal.
  \item Case 2: $G$ is $T$, applied to the $i$th qubit. In
    this case, we let
    \begin{equation}
      A_1 = \m{
        \begin{qcircuit}[scale=0.45]
          \grid{7}{1,3,4,5};
          \gridx{1}{6}{0};
          \draw (0.5,2) node {$\ldots$};
          \init{$\ket 0$}{1,0};
          \wirelabel{$i$}{1,3};
          \controlled{\notgate}{2,0}{3};
          \biggate{$A'_1$}{3,1}{4,5};
          \gate{$T$}{3.5,0};
          \wirelabel{$i$}{6,3};
          \controlled{\notgate}{5,0}{3};
          \term{$\ket 0$}{6,0};
          \draw (6.5,2) node {$\ldots$};
        \end{qcircuit}
      }
    \end{equation}
    and $A_2=A'_2$.  Since $A'_1$ is diagonal, so is $A_1$, and it
    follows that the ancilla is uncomputed correctly.  Moreover, $A_1$
    is equivalent to $A'_1\cp G$, and therefore, $A=A_2\cp
    A_1$. Finally, since $A'_1$ has $T$-depth at most 1, so does
    $A_1$.
  \item Case 3: $G$ is $T\da$, applied to the $i$th qubit. This is
    entirely analogous to case 2.\qedhere
\end{itemize}
\end{proof}

A similar result appears in Section~6.4 of version 2 of
{\cite{AMMR12}}, but with a proof that is quite different. 

Note that the gate set $\PCG$ in Theorem~\ref{thm-almost} is not necessarily
assumed to consist of Clifford gates. For example, if on some
hypothetical architecture, $T$-gates are expensive but Toffoli gates
are cheap, one can include the Toffoli gate in the set $\PCG$.

In general, the proof of Theorem~\ref{thm-almost} increases the size
of the circuit, but only by a constant factor. In practice, it is
often possible to find a much smaller circuit than the one constructed
in the proof.

If we take $\PCG=\s{S,X,\CNOT}$ and apply Theorem~\ref{thm-almost}
to circuit (\ref{eqn-nc}) (excluding the initial and final Hadamard
gate), we obtain another $T$-depth one representation of the Toffoli gate.

We also note that there is a trade-off between $T$-depth and the
number of ancillas.  The procedure of the proof of
Theorem~\ref{thm-almost} adds one ancilla for each $T$-gate. However,
by splitting a circuit with $T$-count $n$ into two circuits with
$T$-count $\ceil{n/2}$ each, it is clear that one can approximately
half the number of ancillas by doubling the $T$-depth, and so forth.

\section{Some circuits cannot be written with $T$-depth one}

The result of the previous section shows that any two $T$-stages can
be combined into a single $T$-stage, provided that they are only
separated by almost classical gates. One may wonder whether perhaps
{\em all} Clifford+$T$ circuits can be written of $T$-depth one, using
a sufficient number of ancillas initialized to $\ket 0$. We show that
this cannot be done.

\begin{theorem}\label{thm-tht}
  The single-qubit operator $THT$ cannot be implemented as a
  Clifford+$T$ circuit of $T$-depth 1, using an arbitrary number of
  ancillas initialized to $\ket 0$. This is true even if the ancillas
  are not required to be returned to their initial state at the end of
  the computation.
\end{theorem}

Before proving the theorem, we start with a general observation about
Clifford+$T$ circuits of $T$-depth 1.

\begin{proposition}\label{prop-t1}
  Let $U$ be an $n$-qubit Clifford+$T$ circuit of $T$-depth 1. Let
  $\ket\phi$ be any single-qubit state, and consider
  \[ \ket{\psi} =
  U(\ket{\phi}\otimes\ket{0}\otimes\ldots\otimes\ket{0}).
  \]
  Consider the $\s{+1,-1}$-valued Pauli observable $X$ applied to the
  first qubit of $\psi$; denote its expected value by $E_{\ket\phi}$.
  Suppose $E_{\ket{+}}$ is non-zero. Then
  \[ \frac{E_{\ket{0}}}{E_{\ket{+}}}
  \]
  is a rational number.
\end{proposition}

\begin{proof}
  The expected value of the observable $X$ on the first qubit of
  $\ket{\psi}$ is 
  \begin{equation}\label{eqn-5.1}
    \begin{split}
      E_{\ket\phi} &= \bra{\psi}\;(X\otimes I\dotted\otimes I)\;\ket{\psi}\\
      &= \bra{\phi,0\dotted,0}\;U\da(X\otimes
      I\dotted\otimes I)U\;\ket{\phi,0\dotted,0}.
    \end{split}
  \end{equation}
  We analyze the structure of $U\da(X\otimes I\dotted\otimes
  I)U$. Since $U$ is of $T$-depth 1, it can be written as $U=U_1\circ
  U_2\circ U_3$, where $U_1$ and $U_3$ are Clifford circuits and $U_2
  = T\dotted\otimes T\otimes I\dotted\otimes I$. Since $U_1$ is
  Clifford, we know that $U_1\da(X\otimes I\dotted\otimes I)U_1$ is a
  Pauli operator
  \begin{equation}\label{eqn-5.2}
    U_1\da(X\otimes I\dotted\otimes I)U_1 = \pm A_1\dotted\otimes A_n,
  \end{equation}
  where each $A_i\in\s{X,Y,Z,I}$. Using the relations
  \begin{align*}
    T\da IT &= I, &
    T\da ZT &= Z, \\
    T\da XT &= \frac{1}{\sqrt{2}}X - \frac{1}{\sqrt{2}}Y, &
    T\da YT &= \frac{1}{\sqrt{2}}X + \frac{1}{\sqrt{2}}Y,
  \end{align*}
  we find that 
  \begin{equation}
    \begin{split}
      &U_2\da(\pm A_1\dotted\otimes A_n) U_2 
      \\&= 
      \pm(T\da A_1 T)\dotted\otimes(T\da A_{n_1} T)\otimes
      A_{n_1+1} \dotted\otimes A_n 
      \\&=
      \lambda P_1+\lambda P_2\dotted + \lambda P_m,
    \end{split}
  \end{equation}
  where each $P_j$ is an $n$-qubit Pauli operator. The key observation
  here is that the {\em same} factor $\lambda$ occurs in front of each
  (possibly signed) summand, and $\lambda$ is independent of
  $\ket{\phi}$. In fact, we have $\lambda=(\frac{1}{\sqrt{2}})^{k}$,
  where $k$ is the number of times the operators $X$ and $Y$ occur
  among $A_1,\ldots,A_{n_1}$. Let
  \begin{equation}
    Q_j = U_3\da\;P_j\; U_3.
  \end{equation}
  Since $U_3$ is Clifford, this is again some Pauli operator, say
  \begin{equation}\label{eqn-5.5}
    Q_j = (-1)^{q_j} B_{j,1}\dotted\otimes B_{j,n}.
  \end{equation}
  Combining (\ref{eqn-5.2}) through (\ref{eqn-5.5}), we find
  \begin{equation}
    \begin{split}
      U\da(X\otimes I\dotted\otimes I)U &=
      \lambda Q_1+\lambda Q_2\dotted + \lambda Q_m
      \\&= \lambda\sum_{j=1}^m (-1)^{q_j} B_{j,1}\dotted\otimes B_{j,n}.
    \end{split}
  \end{equation}
  Combining this with (\ref{eqn-5.1}), we get
  \begin{equation}
    E_{\ket\phi} = \lambda\sum_{j=1}^m (-1)^{q_j}
    \bra{\phi}B_{j,1}\ket{\phi}
    \;\bra{0}B_{j,2}\ket{0}
    \;\cdots\; \bra{0}B_{j,n}\ket{0}.
  \end{equation}
  Since each $B_{j,i}\in\s{X,Y,Z,I}$ is a Pauli operator, it follows
  that $E_{\ket\phi}/\lambda$ is rational (indeed, an integer) for
  $\ket{\phi}\in\s{\ket{0},\ket{+}}$. The claim then immediately follows.
\end{proof}

\begin{proof}[Proof of Theorem~\ref{thm-tht}]
  For $U=THT$, we compute 
  \[ U\da X U = \frac{1}{2} X + \frac{1}{2} Y + \frac{1}{\sqrt{2}} Z,
  \]
  and therefore
  \[ E_{\ket{0}} = \bra{0}\; U\da XU \;\ket{0} = \frac{1}{\sqrt{2}}
  \]
  and
  \[
  E_{\ket{+}} = \bra{+}\; U\da XU\;\ket{+} = \frac{1}{2}.
  \]
  Since $E_{\ket{0}}/E_{\ket{+}}$ is irrational, the claim immediately
  follows from Proposition~\ref{prop-t1}.
\end{proof}

\section{Conclusion}

We found a class of circuits whose $T$-depth can be reduced to one, by
using a sufficient number of ancillas. We also showed that there are
circuits whose $T$-depth cannot be reduced to one, regardless of the
number of ancillas used. It remains an open problem how to determine
the minimal $T$-depth or $T$-count of any given Clifford+$T$ circuit.

\section{Acknowledgements}

This research was supported by the Natural Sciences and Engineering
Research Council of Canada (NSERC). This research was also supported
by the Intelligence Advanced Research Projects Activity (IARPA) via
Department of Interior National Business Center contract number
D11PC20168. The U.S. Government is authorized to reproduce and
distribute reprints for Governmental purposes notwithstanding any
copyright annotation thereon. Disclaimer: The views and conclusions
contained herein are those of the authors and should not be
interpreted as necessarily representing the official policies or
endorsements, either expressed or implied, of IARPA, DoI/NBC, or the
U.S. Government.

\bibliographystyle{abbrvunsrt} 
\bibliography{toffoli}

\end{document}